\documentclass[runningheads]{llncs}
\usepackage[utf8]{inputenc}
\usepackage{amsmath,amssymb}
\usepackage[classfont=typewriter]{complexity}
\usepackage{xspace}
\usepackage{paralist}
\usepackage{todonotes}
\usepackage{hyperref}
\usepackage{cleveref}
\usepackage{subfig}
\usepackage{graphicx}
\usepackage{thmtools} 
\usepackage{thm-restate}
\usepackage{enumerate}
\usepackage{soul}
\usepackage{xcolor}

\Crefname{figure}{Fig.}{Figs.}
\Crefname{observation}{Observation}{Observations}

\renewcommand{\orcidID}[1]{\href{https://orcid.org/#1}{\includegraphics[scale=.03]{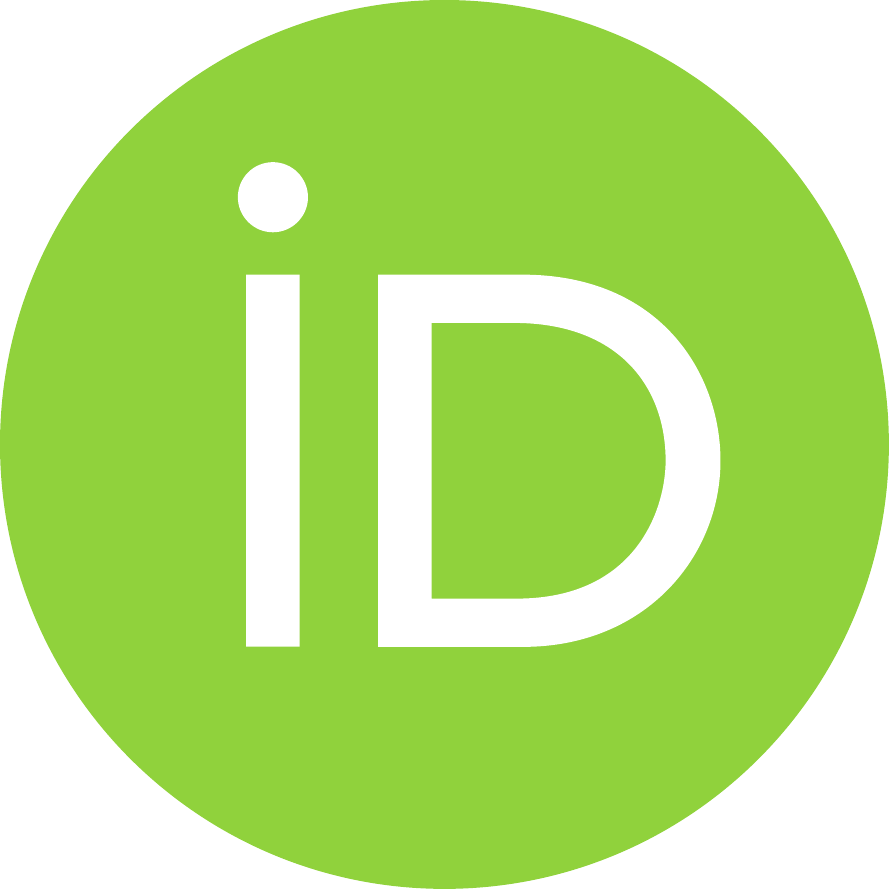}}} 

\title{On the Parameterized Complexity of the\\ \clubedel{s} Problem}

\titlerunning{On the Parameterized Complexity of \clubedel{s}}

\newcommand{\ethshort}{$\mathtt{ETH}$\xspace}
\newcommand{\cedit}{\textsc{Cluster Editing}\xspace}
\newcommand{\cedel}{\textsc{Cluster Edge Deletion}\xspace}

\newcommand{\clubedel}[1]{\textsc{${#1}$-Club Cluster Edge Deletion}\xspace}

\newcommand{\clubvdel}[1]{\textsc{${#1}$-Club Cluster Vertex Deletion}\xspace}

\newcommand{\tw}{\textnormal{tw}}

\newcommand{\bnd}[1]{{\partial{#1}}}
\newcommand{\intr}[1]{{\text{int(}{#1}\text{)}}}

\author{
Fabrizio Montecchiani\orcidID{0000-0002-0543-8912}
\and Giacomo Ortali\orcidID{0000-0002-4481-698X}
\and \\Tommaso Piselli
\and Alessandra Tappini\orcidID{0000-0001-9192-2067}}
\authorrunning{F. Montecchiani et al.}

\institute{Dipartimento di Ingegneria, Università degli Studi di Perugia, Italy 
\email{\{fabrizio.montecchiani,giacomo.ortali,alessandra.tappini\}@unipg.it, tommaso.piselli@studenti.unipg.it}}

\begin{document}

\maketitle

\begin{abstract}
We study the parameterized complexity of the \clubedel{s} problem: Given a graph $G$ and two  integers $s \ge 2$ and $k \ge 1$, is it possible to remove at most $k$ edges from $G$ such that each connected component of the resulting graph has diameter at most~$s$? This problem is known to be \NP-hard already when $s = 2$. We prove that it admits a fixed-parameter tractable algorithm when parameterized by $s$ and the treewidth of the input graph.
\end{abstract}

\section{Introduction}
Graph clustering~\cite{DBLP:journals/csr/Schaeffer07} is a classical task in data mining, with important applications in numerous fields including computational biology~\cite{DBLP:journals/jcb/Ben-DorSY99}, image processing~\cite{DBLP:journals/pami/WuL93}, and machine learning~\cite{DBLP:journals/ml/BansalBC04}. At high-level, the task is to group the vertices of the graph into clusters in such a way that there should be many edges within each cluster and relatively few between different clusters. A prominent formalization  is \cedit (also known as \textsc{Correlation Clustering}): Given a graph $G$ and an integer $k$ as input, the goal is to find a sequence of $k$ operations, each of which can be an edge or vertex insertion or removal, such that the resulting graph is a so-called \emph{cluster graph}, i.e., each of its connected components is a clique. If we restrict the editing operations to be edge removals only,  then the problem is known as \cedel. Namely, \cedel takes a graph $G$ and an integer $k$ as input, and asks for a set of $k$ edges whose removal yields a cluster graph. Equivalently, we seek for a partition of the vertices of $G$ into cliques, such that the  inter-cluster edges (whose end-vertices belong to different cliques) are at most $k$.

Unfortunately, both \cedit and  \cedel are well-known to be \NP-complete in general~\cite{DBLP:journals/dam/ShamirST04},  even when the input instances have bounded vertex degree~\cite{DBLP:journals/dam/KomusiewiczU12}. Indeed, their parameterized complexity with respect to the natural parameter $k$ has been intensively investigated; in particular, both problems are in \FPT~\cite{DBLP:journals/jda/Bocker12,DBLP:journals/jcss/ChenM12}, but do not allow subexponential-time parameterized algorithms unless \ethshort fails~\cite{DBLP:journals/jcss/FominKPPV14,DBLP:journals/dam/KomusiewiczU12}. 

In many applications, modelling clusters with cliques might be a severe limitation, for instance, in presence of noise in the data collection process. Consequently, several notions of relaxed cliques have been introduced and investigated~\cite{Balasundaram2013,DBLP:journals/algorithms/Komusiewicz16}. We focus on the concept of $s$-club, in which each pair of vertices is at distance at most $s \ge 2$ in the cluster. (Note that a $1$-club is in fact a clique). We remark that defining clusters as $s$-clubs proved to be effective in several application scenarios such as social network analysis and bioinformatics~\cite{alba,DBLP:journals/jco/BalasundaramBT05,laan,mokken,mokken2}. The \clubedel{s} problem can be stated analogously as \cedel by replacing cliques with $s$-clubs (see \Cref{se:preliminaries} for formal definitions). Unfortunately, \clubedel{s} is \NP-complete already for $s=2$~\cite{DBLP:conf/aaim/LiuZZ12}. Also, \clubedel{2} belongs to \FPT~parameterized by $k$~\cite{DBLP:journals/corr/abs-2107-01133,DBLP:conf/aaim/LiuZZ12}, and it admits no subexponential-time parameterized  algorithm in $k$~\cite{DBLP:journals/jcss/MisraPS20}. More in general, for any $s \ge 2$, \clubedel{s} cannot be solved in time $2^{o(k)}n^{O(1)}$ unless \ethshort fails~\cite{DBLP:journals/jcss/MisraPS20}.

Based on the above discussion, we know that it is unlikely that \clubedel{s} lies is \FPT~when parameterized by $s$, whereas the complexity of the problem parameterized by $s+k$ is open to the best of our knowledge. In this paper, we instead focus on those scenarios in which the solution size (measured by $k$) is large, and we still aim for tractable problems based on alternative parameterizations. In this respect,  treewidth is a central parameter in the parameterized complexity analysis (see~\cite{DBLP:series/mcs/DowneyF99,DBLP:journals/jal/RobertsonS86}). We prove that \clubedel{s} lies in  \FPT~when parameterized by $s+\tw$, where $\tw$ is an upper bound for the treewidth of the input graph. More precisely, our main contribution can be summarized as follows.

\medskip\begin{theorem}\label{th:main}
Let $G$ be an $n$-vertex graph of treewidth at most  $\tw$. There is an algorithm that solves the \clubedel{s} problem on $G$ in $O(2^{2^{O(\tw^2 \log s)}}\cdot n)$ time.
\end{theorem}

 From the technical point of view, the main crux of our approach lies in the definition of sufficiently small records that allow to keep track of the distances between pairs of vertices in a (partial) $s$-club. With such records at hand, we then apply a standard DP algorithm over a tree decomposition of the input graph, which still requires a nontrivial amount of technicalities in order to update the records. Our records have some similarities, but also several key differences and extensions, with those used in a technique presented by Dondi and Lafond in~\cite[Thm. 14]{DBLP:conf/fct/DondiL19} to solve a different problem. Namely, they describe a fixed-parameter tractable algorithm for the problem of deciding whether the vertices of a graph can be covered with at most $d$ (possibly overlapping)  $2$-clubs, parameterized by treewidth. The main novelties of the presented approach with respect to~\cite{DBLP:conf/fct/DondiL19} will be suitably highlighted throughout the paper.

\section{Preliminaries}\label{se:preliminaries}
For any $d \in \mathbb{Z}^+$, we use $[d]$ as shorthand for the set $\{1,2,\dots,d\}$. Let $G=(V,E)$ be a graph. For any $W \subset V$, we denote by $G[W]$ the subgraph of $G$ induced by the vertices of $W$. The \emph{neighborhood} of a vertex $v$ of $G$ is defined as $N_G(v)=\{u: uv\in E \}$. Given two vertices $u,v \in V$, the \emph{distance in $G$} between $u$ and $v$, denoted by $d_G(u,v)$, is the number of edges in any shortest path between $u$ and $v$ in $G$. The \emph{diameter} of $G$ is the maximum distance in $G$ between any two of its vertices.  An \emph{$s$-club} of $G$, with $s \ge 1$, is a subset $W \subseteq V$ such that the diameter of $G[W]$ is at most $s$. A \emph{partition} of $G$ is a collection of subsets $\mathcal{C} = \{C_i\}_{i \in [d]}$ such that: (a) $\bigcup_{i=1}^{d}C_i = V$, and (b) $C_i \cap C_j = \emptyset$ for each $i,j \in [d]$ with $i \neq j$. We denote by $E_\mathcal{C}$ the set of all edges $uv$ of $G$ such that $u,v \in C_i$, for some $i \in [d]$. 
We are interested in the following problem.

\medskip\noindent\fbox{%
  \parbox{0.95\textwidth}{
\clubedel{s}\\
\textnormal{\textbf{Input:}} $G=(V,E)$, $k \ge 1$, $s \ge 2$.\\
\textnormal{\textbf{Output:}} A partition $\mathcal{C}= \{C_i\}_{i \in [d]}$ of $G$ such that $C_i$ is an $s$-club for each $i \in [d]$, and  $|E \setminus E_\mathcal{C}| \le k$. 
  }%
}

\medskip Let $(\mathcal{X},T)$ be a pair such that $\mathcal{X}=\{X_i\}_{i \in [\ell]}$ is a collection of subsets of vertices of a graph $G=(V,E)$, called \emph{bags}, and $T$ is a tree whose nodes are in  one-to-one correspondence with the elements of $\mathcal X$. When this creates no ambiguity, $X_i$ will denote both a bag of $\mathcal{X}$ and the node of $T$ whose corresponding bag is $X_i$. The pair $(\mathcal{X},T)$ is a \emph{tree-decomposition} of $G$ if:
\begin{inparaenum}[(i)]
\item $\bigcup_{i\in [\ell]} X_i = V$,
\item for every edge $uv$ of $G$, there exists a bag $X_i$ that contains both $u$ and~$v$, and
\item for every vertex $v$ of $G$, the set of nodes of $T$ whose bags contain $v$ induces a non-empty (connected) subtree of $T$.
\end{inparaenum}
The \emph{width} of  $(\mathcal{X},T)$ is $\max_{i=1}^\ell {|X_i| - 1}$, while the \emph{treewidth} of $G$, denoted by $\tw(G)$, is the minimum width over all tree-decompositions of $G$. 
 A tree-decomposition $(\mathcal{X},T)$ of a graph $G$ is \emph{nice} if $T$ is a rooted tree with the following additional  properties~\cite{DBLP:journals/jal/BodlaenderK96}.
 \begin{inparaenum}[(P.1)]
 \item Every node of $T$ has at most two children.
 \item If a node  $X_i$ of $T$ has two children whose bags are $X_j$ and $X_{j'}$, then $X_i=X_j=X_{j'}$. In this case, $X_i$ is a \emph{join bag}.
 \item If a node $X_i$ of $T$ has only one child~$X_j$, then $X_i \neq X_j$ and there exists a vertex $v \in G$ such that either $X_i = X_j \cup \{v\}$ or $X_i \cup \{v\} = X_j$. In the former case  $X_i$ is an \emph{introduce bag}, while in the latter case  $X_i$ is a \emph{forget bag}.
 \item If a node $X_i$ is the root or a leaf of $T$, then $X_i=\emptyset$.\end{inparaenum} 
~Given a tree-decomposition of width $w$ of $G$, a nice tree-decomposition of width $w$ can be computed in $O(w \cdot n)$ time~\cite{DBLP:books/sp/Kloks94}.

\section{Proof of \Cref{th:main}}

The proof is based on a DP algorithm over a nice tree-decomposition. We first describe the records to be stored at each~bag, and we then present~the~algorithm.

\subsection{Definition of the records}
Let $G=(V,E)$ be an $n$-vertex graph and let $(\mathcal{X},T)$ be a nice tree-decomposition of $G$ of width $\tw$. For each $i \in [\ell]$, let $T_i$ be the subtree of $T$ rooted at the bag $X_i \in \mathcal{X}$ and let $G_i=(V_i,E_i)$ be the subgraph of $G$ induced by the vertices that belong to at least one bag of $T_i$. Moreover, given a subset of vertices $C \subseteq V_i$, we call $C$ a \emph{potential $s$-club}, and let $\bnd{C}=C \cap X_i$  and $\intr{C}=C \setminus X_i$. We are now ready to describe the items of the record to be stored at each bag. Some of these items are similar to those described in~\cite{DBLP:conf/fct/DondiL19}, although with some important differences that allow us to deal with any fixed value of $s$. 

    The first item of the record is a table that stores the pairwise distances of the vertices in $\bnd{C}$. Namely, let $D(\bnd{C})$ be a table having one row and one column for each vertex in $\bnd{C}$, and such that:
$$ 
D(\bnd{C})[a,b]=
\begin{cases}
0,  &\text{if}~a=b\\
d_{G[C]}(a,b), &\text{if}~1 \le d_{G[C]}(a,b) \le s\\
\infty, &\text{otherwise}.
\end{cases}
$$
     Observe that  $D(\bnd{C})$ contains at most $(\tw+1)^2 \in O(\tw^2)$ entries.
    
    \smallskip The second item of the record is a table that stores the distance between pairs of vertices such that one is in $\bnd{C}$ and the other is in $\intr{C}$. Two vertices $u,u'$ in $\intr{C}$ are \emph{equivalent} with respect to $\bnd{C}$, if for each vertex $a \in \bnd{C}$, then either $1 \le d_{G[C]}(u,a)=d_{G[C]}(u',a) \le s$, or $d_{G[C]}(u,a) > s$ and $d_{G[C]}(u',a) > s$.     Namely, let $H(\bnd{C})$ be a table having one column for each vertex $a \in \bnd{C}$, and one row for each equivalence class $[u]_{\bnd{C}}$ with respect to $\bnd{C}$. In particular, we have: 
$$ 
H(\bnd{C})[u,a]=
\begin{cases}
d_{G[C]}(u,a),  &\text{if}~1 \le d_{G[C]}(u,a) \le s\\
\infty, &\text{otherwise}.
\end{cases}
$$    
    
    Observe that  $H(\bnd{C})$ contains at most $(s+1)^{(\tw+1)}$ rows, and hence $O(\tw \cdot s^{O(\tw)}) = O(2^{O(\tw \log{s})})$ entries. It is worth observing that, if $H(\bnd{C})$ contains two rows $r$ and $r'$ such that cell-wise the values of $r$ are smaller than or equal to those of $r'$, then we can avoid storing $r$ and keep only $r'$ in $H(\bnd{C})$. However, this would not reduce the asymptotic size of $H(\bnd{C})$.
    
    \smallskip The third (and last) item of the record represents the key insight to extend the result in~\cite{DBLP:conf/fct/DondiL19}. Roughly speaking, in the case $s=2$, pairs of vertices in $\intr{C}$ must have distance at most two in $G[C]$, otherwise there is no $C' \supseteq C$ such that $C'$ is a $2$-club of $G$. Unfortunately, this is not true in general when $s>2$. On the other hand, suppose that $C$ is a subset of an $s$-club $C'$ of $G$ and that there exist two vertices $u,u' \in \intr{C}$ whose distance in $G[C]$ is larger than $s$. Then, by the properties of a tree-decomposition, the shortest path in $G$ between $u$ and $u'$ goes through some pair of vertices in $\bnd{C}$. We formalize this observation. Let $w,z \in \intr{C}$ be two vertices such that $d_{G[C]}(w,z)>s$. A \emph{request for $\bnd{C}$}, denoted by $R_{wz}$, is a table having one row and one column for each vertex in $\bnd{C}$. Namely, for each $a,b \in \bnd{C}$, if there exists $2 \le \delta \le s-2$ such that connecting $a$ and $b$ with a path $\pi$ of length $\delta$ makes the distance between $w$ and $z$ to be at most $s$, then $R_{wz}[a,b]=\delta$, while $R_{wz}[a,b]=\star$ otherwise.  Observe that if there exist two requests $R_{wz}$ and $R_{w'z'}$ such that $R_{wz}[a,b]=R_{w'z'}[a,b]$ for each pair $a,b \in \bnd{C}$, then $w$ and $w'$ are equivalent with respect to $\bnd{C}$ (i.e., $w,w' \in [w]_{\bnd{C}}$), and the same holds for $z$ and $z'$. Therefore we can avoid storing duplicated requests, and we denote by $Q(\bnd{C})$ the set containing all distinct requests for~$\bnd{C}$. Also, $Q$ contains at most $(s-2)^{(\tw+1)^2} \in O(s^{O(\tw^2)}) = O(2^{O(\tw^2 \log{s})})$ distinct requests. 

\smallskip Before describing our algorithm, we need to slightly extend our notation. If a potential $s$-club $C$ is such that $\bnd{C} = \emptyset$ (recall that $C \subseteq V_i$), then we call it \emph{complete}. Consider a partitioning $\mathcal{P}^l_i$ of $G_i$ into potential $s$-clubs and let  $\mathcal{C}^l_i=\{C^l_{j,i} \mid j \in [d_l]\}$ be the potential $s$-clubs in $\mathcal{P}^l_i$ that are not complete, i.e., any $C \in \mathcal{C}^l_i$ is such that $\bnd{C} \neq \emptyset$. In particular, let  $\bnd{\mathcal{C}^l_i}=\{\bnd{C^l_{j,i}} \mid j \in [d_l]\}$. Moreover, we let $\mathcal{D}^l_i= \{D(\bnd{C^l_{j,i}}) \mid j \in [d_l]\}$, $\mathcal{H}^l_i= \{H(\bnd{C^l_{j,i}}) \mid  j \in [d_l]\}$, and $\mathcal{Q}^l_i= \{Q(\bnd{C^l_{j,i}}) \mid j \in [d_l]\}$. A \emph{solution} of $X_i$ is a tuple $S^l_i=\langle \bnd{\mathcal{C}^l_i}, \mathcal{D}^l_i, \mathcal{H}^l_i, \mathcal{Q}^l_i, k^l_i  \rangle$. Here  $k^l_i$ is an integer, called \emph{edge-counter}, equal to $|E_i \setminus \mathcal{P}_i^l(E_i)|$, i.e., $k^l_i$ counts the number of edges having their endpoints in different potential $s$-clubs of $\mathcal{P}_i^l$, 
hence $k^l_i \le k$. Two solutions $S^l_i=\langle \bnd{\mathcal{C}^l_i}, \mathcal{D}^l_i, \mathcal{H}^l_i, \mathcal{Q}^l_i,  k^l_i  \rangle$ and $S^{g}_i=\langle \bnd{\mathcal{C}^g_i}, \mathcal{D}^g_i, \mathcal{H}^g_i, \mathcal{Q}^g_i, k^g_i  \rangle$ are \emph{distinct} if $\bnd{\mathcal{C}^l_i} \neq \bnd{\mathcal{C}^g_i}$, or $\mathcal{D}^l_i \neq \mathcal{D}^g_i$, or $\mathcal{H}^l_i \neq \mathcal{H}^g_i$, or $\mathcal{Q}^l_i \neq \mathcal{Q}^g_i$. Observe that if $S^l_i$ and $S^g_i$ are not distinct but $k^l_i < k^g_i$, then it suffices to consider only $S^l_i$. 

\begin{lemma}\label{le:number-solutions}
For a bag $X_i$, there exist $O(2^{2^{O(\tw^2 \log s)}})$ distinct solutions.
\end{lemma}
\begin{proof}
The number of distinct solutions is at most $n_C \times n_D \times n_H \times n_Q$, where $n_C$ is the number of possible partitions of the vertices in $X_i$, while $n_D$, $n_H$, $n_Q$ are the number of possible sets $\mathcal{D}^l_i$, $\mathcal{H}^l_i$, and $\mathcal{Q}^l_i$, respectively, for a fixed partition. Thus, since $n_C \in O(2^{O(\tw)})$, $n_D \in O(2^{O(\tw^2 \log{s})})$, $n_H \in O(2^{2^{O(\tw \log s)}})$, and $n_Q \in O(2^{2^{O(\tw^2 \log s)}})$, the statement follows.
\qed\end{proof}

\subsection{Description of the algorithm}

We are now ready to describe our DP algorithm over $(\mathcal{X},T)$. The main differences with respect to the algorithm in~\cite{DBLP:conf/fct/DondiL19} lie in the management of the sets $Q(\cdot)$ (which do not exist in~\cite{DBLP:conf/fct/DondiL19}), and on a more sophisticated updating of the tables $D(\cdot)$ and $H(\cdot)$, as a consequence of the non applicability of some simplifying assumptions that hold only when $s=2$. Moreover, we also keep track of the number of edges having their end-vertices in different potential $s$-clubs.

Let $X_i$ be the current bag visited by the algorithm. We compute the set of solutions for $X_i$ based on the solutions computed for its child or children (if any). In case the resulting set of solutions is empty, the algorithm halts and returns a negative answer. We distinguish four cases based on the type of $X_i$.

\medskip\noindent\textbf{$X_i$ is a leaf bag.} In this case $X_i$ is empty and there exists only one trivial  solution $S^1_i$, in which all tables are empty and $k^1_i=0$.

\medskip\noindent\textbf{$X_i$ is an introduce bag.} Let  $X_j =X_i  \setminus \{v\}$ be the child of $X_i$.  The algorithm exhaustively extends each solution $S^l_j$ of $X_j$ as follows. It first generates at most $d_l$ new partitions by placing $v$ in each $\bnd{C'} \in \bnd{\mathcal{C}^l_j}$. Also, it generates a partition in which $v$ forms a new potential $s$-club $C = \bnd{C} = \{v\}$. Consider one of the new partitions generated by the algorithm. 
In order to build the corresponding new solution for $X_i$, we distinguish the following two cases.

\medskip\noindent\textbf{Case A ($\bnd{C}=\{v\}$).}  $D(\bnd{C})$ is trivially defined,  $H(\bnd{C})$ and $Q(\bnd{C})$ are empty. 

\medskip\noindent\textbf{Case B ($\bnd{C} = \bnd{C'} \cup \{v\}$).}
The next observation immediately follows from the fact that
 $\bnd{C} =\bnd{C'} \cup \{v\}$ and $\intr{C}=\intr{C'}$.
 
 \begin{observation}\label{obs:introduce-distance}
 Suppose that there exist $a,b \in \bnd{C'}$ such that $d_{G[C']}(a,b) > d_{G[C]}(a,b)$, then any shortest path between $a$ and $b$ in $G[C]$ contains vertex~$v$.
 \end{observation}
 
 
\noindent -- Computing $D(\bnd{C})$ from $D(\bnd{C'})$\footnote{Since the matrix is symmetric, when we update a cell $D(\bnd{C})[a,b]$ we assume that also $D(\bnd{C})[b,a]$ is updated with the same value.}.
        \begin{enumerate}
            \item We add a new row and a new column for vertex $v$. 
            \item For each vertex $a \in \bnd{C'}$, let $\delta_{av}=\min_{b \in N_{G[X_i]}(v)}D(\bnd{C'})[a,b]$, and note that $\delta_{av}=0$ if edge $av$ belongs to $G[C]$. Clearly, it holds that
            $$D(\bnd{C})[a,v]=
            \begin{cases}
            \infty, &\text{if}~\delta_{av} \in \{s,\infty\}\\
            1+\delta_{av}, &\text{otherwise.}
            \end{cases}$$
            
            Hence, we have $D(\bnd{C})[u_2,v]=3$.
            \item By \Cref{obs:introduce-distance}, for each pair $a,b \in \bnd{C'}$, the corresponding value of $D(\bnd{C})$ can be updated as follows:  
            $$D(\bnd{C})[a,b] =  \min\{D(\bnd{C'})[a,b],D(\bnd{C})[a,v]+D(\bnd{C})[b,v]\}.$$

            
        \end{enumerate}
        
        
\noindent -- Computing $H(\bnd{C})$ from $H(\bnd{C'})$. 
        \begin{enumerate}
            \item We  add a new column for vertex $v$.
            \item For each equivalence class $[u]_{\bnd{C'}}$, let $\delta_{uv}=\min_{a \in N_{G[X_i]}(v)}{H(\bnd{C'})[u,a]}$. Since there is no edge $uv$ such that $u \in \intr{C}$, it follows that
            $$H(\bnd{C})[u,v]=
            \begin{cases}
            \infty, &\text{if}~\delta_{uv} \in \{s, \infty\}\\
            1+\delta_{uv}, &\text{otherwise.}
            \end{cases}$$
            \item By \Cref{obs:introduce-distance}, for each pair of vertices $u \in  \intr{C'}$ and $a \in \bnd{C'}$,  the corresponding value of $H(\bnd{C})$ can be updated as follows: 
            $$H(\bnd{C})[u,a] = \min\{H(\bnd{C'})[u,a],H(\bnd{C})[u,v]+D(\bnd{C})[v,a]\}.$$ 
        \end{enumerate}
\noindent -- Computing $Q(\bnd{C})$ from $Q(\bnd{C'})$. Note that the addition of $v$ cannot lead to new requests but it may actually yield the update of some  request in $Q(\bnd{C'})$.
        \begin{enumerate}
            \item For each request $R_{wz}$ in $Q(\bnd{C'})$, we verify whether, as a consequence of the introduction of $v$, there exists a cell $R_{wz}[a,b]$ such that $D(\bnd{C})[a,b] \le R_{wz}[a,b]$. If such a cell exists,  we say that $R_{wz}$ is \emph{fulfilled}. We add $R_{wz}$ to $Q(\bnd{C})$ if and only if $R_{wz}$ is  not fulfilled. 
            \item If $R_{wz}$ is not fulfilled, before adding it to $Q(\bnd{C})$, we update it as follows:
            \begin{enumerate}
                \item We add a row and a column for $v$.
                \item For each pair $a,b \in \bnd{C'}$, we compute
                \begin{multline*}
                    \delta_{ab}= \min\{(H(\bnd{C})[w,a]+H(\bnd{C})[z,b], H(\bnd{C})[z,a]+H(\bnd{C})[w,b]\}.
                \end{multline*}
                Observe that $\delta_{ab} + D(\bnd{C})[a,b]>s$, otherwise the request would have been fulfilled before.
                \item By definition of request, we have $R_{wz}[a,b]=s-\delta_{ab}$, if $\delta_{ab} < s-1$, and  $R_{wz}[a,b]=\star$, otherwise.
            \end{enumerate}
        \end{enumerate}
    
Finally, in both \textbf{Case A} and \textbf{Case B}, we observe that, in order to obtain the edge-counter of the new solution, $k^l_j$ needs to be increased by the number of edges incident to $v$ whose other end-vertex is in $X_i$ but not in $C$. If the resulting edge-counter is greater than $k$, the solution is discarded. 

\medskip\noindent\textbf{$X_i$ is a forget bag.} Let $X_j = X_i \cup \{v\}$ be the child of $X_i$. The algorithm updates each solution $S^l_j$ of $X_j$ as follows. It first identifies the potential $s$-club $C' \in \bnd{\mathcal{C}^l_j}$ that vertex $v$ belongs to. Then, it verifies whether $\bnd{C'}=\{v\}$, i.e., whether removing $v$ from $C'$ makes it complete. 

\medskip\noindent\textbf{Case A ($\bnd{C'}=\{v\}$).} The algorithm verifies the following \emph{completion conditions}: (i) The value of each cell of $D(\bnd{C'})$ is at most $s$; (ii) The value of each cell of $H(\bnd{C'})$ is at most $s$; (iii) The set $Q(\bnd{C'})$ is empty. We observe the following.

\begin{observation}\label{obs:complete}
If $C'$ is complete, then it is an $s$-club of $G$ if and only if the completion conditions are satisfied.
\end{observation}

If any of the completion conditions is not satisfied,  the solution $S^l_j$ is discarded. Otherwise,  we generate a new solution for $X_i$ such that $\bnd{\mathcal{C}^l_i} = \bnd{\mathcal{C}^l_j} \setminus\{\bnd{C'}\}$.

\medskip\noindent\textbf{Case B ($\bnd{C} \supset \{v\}$).}  First, we set $\bnd{C}= \bnd{C'} \setminus \{v\}$ and $\bnd{\mathcal{C}^l_i} = (\bnd{\mathcal{C}^l_j} \setminus\{\bnd{C'}\}) \cup \{\bnd{C}\}$. Note that  the distance in $G[C]$ between any two vertices is the same as in $G[C']$.  
    
\noindent -- Computing $D(\bnd{C})$ from $D(\bnd{C'})$: We  remove the row and column of $v$. 

\noindent --  Computing $H(\bnd{C})$ from $H (\bnd{C'})$: 
        \begin{enumerate}
            \item We remove the column corresponding to $v$.
            \item We check whether $H(\bnd{C'})$ already contains a row that represents the distances of $v$ with respect to the vertices in $\bnd{C}$. Namely, we check if there is an equivalence class $[u]_{\bnd{C}}$ such that for each vertex $a \in \bnd{C}$, it holds $H(\bnd{C})[u,a]=D(\bnd{C'})[a,v]$. If such a row does not exist, we add it to $H(\bnd{C})$. 
        \end{enumerate}
        
\noindent -- Computing $Q(\bnd{C})$ from $Q(\bnd{C'})$. Forgetting vertex $v$  causes the update of existing requests in $Q(\bnd{C'})$, as well as the introduction of new requests. 
        \begin{enumerate}
            \item To update the existing requests, for each request $R_{wz}$ in $Q(\bnd{C'})$:
            \begin{enumerate}
                \item We remove the row and column corresponding to $v$.
                \item We  verify that there exist at least two vertices $a,b \in \bnd{C}$ such that $R_{wz}[a,b] \neq \star$. If this is the case, we add the updated request to $Q(\bnd{C})$, otherwise we discard the solution $S^l_j$. 
            \end{enumerate}
            \item To introduce new requests, we  verify whether there exists a cell $H(\bnd{C'})[u,v]=\infty$, which represents the existence of an equivalence class $[u]_{\bnd{C'}}$ whose vertices have distance in $C'$ from $v$ larger than $s$. If such a cell exists, then:
            \begin{enumerate}
                \item We create a new request $R_{uv}$.
                \item  For each pair $a,b \in \bnd{C}$, we compute $$\delta_u = \min\{H(\bnd{C})[u,a],H(\bnd{C})[u,b]\}$$ and $$\delta_v = \min\{H(\bnd{C})[v,a],H(\bnd{C})[v,b]\}.$$
                \item By definition of request, we have $R_{uv}[a,b]=s-(\delta_u+\delta_v)$, if $\delta_u+\delta_v < s-1$, and  $R_{uv}[a,b]=\star$, otherwise.
                \item If the value of at least one cell of $R_{uv}$ is different from $\star$, then we add $R_{uv}$ to $Q^l_i(\bnd{C})$. Otherwise, we discard the solution $S^l_j$.
            \end{enumerate}
        \end{enumerate}
        

In both \textbf{Case A} and \textbf{Case B}, we observe that the edge-counter of the new solution can be set to be equal to the original $k^l_j$.

Finally, we observe that two solutions $S^l_i$ and $S^g_i$, stemming from two distinct solutions of $X_j$, may now be the same as a consequence of the removal of $v$, up to the values of $k^l_i$ and $k^g_i$. For each such a pair, it suffices to keep the solution with lower edge-counter.

\medskip\noindent\textbf{$X_i$ is a join bag.} Let $X_j = X_{j'}$ be the two children of $X_i$. The algorithm exhaustively merges each pair of solutions $S^l_j$ of $X_j$ and  $S^{l'}_{j'}$ of $X_{j'}$, if possible. A successful merge corresponds to a solution  of $X_i$. Without loss of generality, we can avoid merging $S^l_j$  and  $S^{l'}_{j'}$ when $\bnd{\mathcal{C}^l_j} \neq \bnd{\mathcal{C}^{l'}_{j'}}$, because a resulting solution   (if any), can be obtained by merging a different pair of solutions $S^h_j$ of $X_j$ and  $S^{h'}_{j'}$ such that $\bnd{\mathcal{C}^h_j} = \bnd{\mathcal{C}^{h'}_{j'}}$. Therefore we assume $\bnd{\mathcal{C}^l_j} = \bnd{\mathcal{C}^{l'}_{j'}}$. In other words, for each  $\bnd{C}$ in $\bnd{\mathcal{C}^l_j}$ there exists $\bnd{C'}$ in $\bnd{\mathcal{C}^{l'}_{j'}}$ such that $\bnd{C}=\bnd{C'}$. Also, in the following we denote by $C^*$ the potential $s$-club such that $\bnd{C^*}=\bnd{C}=\bnd{C'}$ and $\intr{C^*}=\intr{C} \cup \intr{C'}$. It remains to verify that, for each such $C^*$, each pair of vertices $u,u'$ such that $u \in \intr{C}$ and $u' \in \intr{C'}$ is either at distance at most $s$ or we can generate a new request for $u,u'$. We remark that, when $s>2$ (and hence differently from \cite{DBLP:conf/fct/DondiL19}), a new shortest path between $u$ and $u'$ may be formed by both vertices in $\intr{C}$ and vertices in $\intr{C'}$. Namely, we proceed as follows. 
\begin{itemize}
    \item For each pair $a,b \in \bnd{C}$, let $\omega_{ab}=\min\{D(\bnd{C})[a,b],D(\bnd{C'})[a,b]\}$. We construct the weighted complete graph $W^*$  on the vertex set $\bnd{C^*}$ and such that the weight of any edge $ab$ is $\omega_{ab}$. 
    \item Computing $D(\bnd{C^*})$ from $D(\bnd{C})$ and $D(\bnd{C'})$: By construction of $W^*$, it follows that for each pair $a,b \in \bnd{C^*}$, $D(\bnd{C^*})[a,b]$ corresponds to the weighted shortest path between $a$ and $b$ in $W^*$.
    
    \item Computing $H(\bnd{C^*})$ from $H(\bnd{C})$ and $H(\bnd{C'})$.
    
    \begin{enumerate}
        \item We first merge $H(\bnd{C})$ and $H(\bnd{C'})$  avoiding duplicated rows. Let $H(\bnd{C''})$ be the resulting table. 
        \item Similarly as in the previous step, for each equivalence class $[u]_{\bnd{C''}}$ of $H(\bnd{C''})$, we add a vertex $u$ to $W^*$ and, for each vertex $a$ of $W^*$ we add the edge $ua$ with weight  $\omega_{ua}=H(\bnd{C''})[u,a]$. Then  $H(\bnd{C^*})[u,a]$ corresponds to the weighted shortest path between $u$ and $a$ in the resulting graph. 
        \item As some rows of $H(\bnd{C^*})$ may be the same,  we remove possible duplicates.
    \end{enumerate}
    
    \item Computing $Q(\bnd{C^*})$ from $Q(\bnd{C})$ and $Q(\bnd{C'})$.
    
    \begin{enumerate}
        \item   We first merge the two sets $Q(\bnd{C})$ and 
    $Q(\bnd{C'})$ avoiding duplicated requests. Let $Q(\bnd{C''})$ be the resulting set of requests.  
    \item We  verify whether some of these requests have been fulfilled. 
        \begin{enumerate}
            \item For each request $R_{wz}$ in $Q(\bnd{C''})$, we add $R_{wz}$ to $Q(\bnd{C^*})$ if and only if $R_{wz}$ is  not fulfilled (i.e., there is no cell $R_{wz}[a,b]$ such that $D(\bnd{C^*})[a,b] \le R_{wz}[a,b]$). 
            \item If $R_{wz}$ is not fulfilled, before adding it to $Q(\bnd{C^*})$, we update it as follows:
            \begin{enumerate}
                \item For each pair of vertices $a,b \in \bnd{C^*}$, we compute
                \begin{multline*}
                    \delta_{ab}= \min\{(H(\bnd{C^*})[w,a]+H(\bnd{C^*})[z,b],\\ H(\bnd{C^*})[z,a]+H(\bnd{C^*})[w,b]\}.
                \end{multline*}
                Observe that $\delta_{ab} + D(\bnd{C^*})[a,b]>s$, otherwise the request would have been fulfilled before.
                \item Similarly as in the previous cases, by definition of requests it follows that $R_{wz}[a,b]=s-\delta_{ab}$, if $\delta_{ab} < s-1$, and ${R_{wz}[a,b]=\star}$, otherwise.
                \item If the value of at least one cell of $R_{wz}$ is different from $\star$,  we add $R_{wz}$ to $Q^l_i(\bnd{C^*})$. Otherwise, we discard the pair of solutions $S^l_j$, $S^{l'}_{j'}$.
            \end{enumerate}
        \end{enumerate}
   \item  We now generate new requests, if needed. 
   
   \begin{enumerate}
       \item  For each pair of rows $[w]_{\bnd{C_{j}}} \in H(\bnd{C_{j}})$ and $[z]_{\bnd{C_{j'}}} \in H(\bnd{C_{j'}})$, we add two representative vertices $w$ and $z$ to $W^*$ and, for each vertex $v$ of $W^*$ we add the edges $wv$ and $zv$ with weights  $H(\bnd{C^*})[w,v]$ and $H(\bnd{C^*})[z,v]$, respectively. Let $\sigma_{uv}$ be the shortest path between any two vertices $u,v$ in this graph. 
       \item If $\sigma_{wz}>s$, we generate a new request $R_{wz}$. 
     \begin{enumerate}
                \item For each pair $u,u' \in \bnd{C^*}$, we compute
    $$\delta_{uu'}= \min\{\sigma_{wu}+\sigma_{zu'}, \sigma_{wu'}+\sigma_{zu}\}.$$
                
    \item Again it follows that $R_{wz}[u,u']=s-\delta_{uu'}$, if $\delta_{uu'} < s-1$, and $R_{wz}[u,u']=\star$, otherwise.
      \item If the value of at least one cell of $R_{wz}$ is different from $\star$,  we add $R_{wz}$ to $Q^l_i(\bnd{C^*})$. Otherwise, we discard the pair of solutions $S^l_j$, $S^{l'}_{j'}$.
     \end{enumerate}
   \end{enumerate}
  
    \end{enumerate}
  
\end{itemize}

Finally, it is readily seen that the edge-counter  of the new solution can be set to be the sum of $k^l_i$ and $k^{l'}_i$. If the resulting edge-counter is larger than $k$, then the solution is discarded.

\medskip The next lemma establishes the correctness of the algorithm.

\begin{lemma}\label{le:correctness}
Graph $G$ admits a solution for \clubedel{s} if and only if the algorithm terminates after visiting the root of $T$.
\end{lemma}
\begin{proof}

Suppose first that the algorithm terminates after visiting the root $\rho$ of $T$. Then the algorithm has computed at least one solution $S_\rho=\langle \emptyset, \emptyset, \emptyset, \emptyset, k_\rho \rangle$ for $X_\rho$. Since $X_\rho$ is an empty bag, by definition of solution there exists a partition $\mathcal{P}_\rho$ of $G$  whose potential $s$-clubs are all complete. Also, each potential $s$-club in $\mathcal{P}_\rho$ satisfied the completion conditions in the corresponding forget bag and hence it is an $s$-club by \Cref{obs:complete}. Moreover, the number of edges of $G$ having their end-vertices in different $s$-clubs of $\mathcal{P}_\rho$ is $k_\rho \le k$. 

Suppose now, to derive a contradiction, that $G$ admits a partition $\cal P$ into $s$-clubs that represents a valid solution but the algorithm terminates prematurely. If this is the case, there must be a bag $X_i$ in which a solution $S_i$ of $X_i$ has been discarded but the partition $\mathcal{P}_i$ corresponding to $S_i$ could be extended to $\mathcal{P}$. We distinguish based on the type of bag $X_i$. 

If $X_i$ is an introduce bag, $S_i$ is discarded only if the edge-counter becomes larger than $k$, which means that there exist at least $k+1$ edges whose end-vertices are in distinct potential $s$-clubs of $\mathcal{P}_i$. Since any two distinct potential $s$-clubs of $\mathcal{P}_i$ will be subsets of two distinct $s$-clubs of $\mathcal{P}$, this is a contradiction with the fact that partition $\cal P$  represents a valid solution. 

If $X_i$ is a forget bag, $S_i$ is discarded either if  it contains a complete potential $s$-club $C$ that does not satisfy the completion conditions or if there is a request $R_{wz}$ in $Q(\bnd{C})$ whose  cells all have value $\star$. In the former case, $C$ is not an $s$-club of $G$ by \Cref{obs:complete}, and there is no potential $s$-club $C'$ of $G$ such that $C \subset C'$, because $\bnd{C}=\emptyset$ and thus vertices in $C$ and vertices in $C' \setminus C$ would form two distinct connected components. In the latter case, any two vertices $u,u'$ such that $u \in [w]_{\bnd{C}}$ and $u' \in [z]_{\bnd{C}}$ are at distance greater than $s$ in $G[C]$ and  there is no pair of vertices $a,b \in \bnd{C}$ such that connecting $a$ and $b$ with a path of length $2 \le \delta \le s-2$ would make the distance between $u$ and $u'$ to be at most $s$. On the other hand, an $s$-club $C'$ of $G$ such that $C \subseteq C'$ would imply the existence of such a path, since $\bnd{C}$ separates $C$ and $C'$. In both cases, $\mathcal{P}$ contains a cluster that is not an $s$-club, a contradiction.

If $X_i$ is a join bag, $S_i$ is discarded  either if  there is a potential $s$-club $C$ whose set $Q(\bnd{C})$ contains a request $R_{wz}$ whose  cells all have value $\star$, or if  the edge-counter becomes larger than $k$. In both cases we can derive a contradiction using the same arguments described in the previous paragraphs.
\end{proof}

\begin{proof}[of \Cref{th:main}]
The correctness of the algorithm derives from \Cref{le:correctness}. It remains to argue about its time complexity. A tree decomposition of $G$ of width $\tw$ can be computed in $O(\tw^{\tw^3} \cdot n)$~\cite{DBLP:journals/siamcomp/Bodlaender96} time, and from it a nice tree-decomposition of width $\tw$  can be derived in $O(\tw \cdot n)$ time~\cite{DBLP:books/sp/Kloks94}.

For each bag $X$, by \Cref{le:number-solutions} and by the fact that we avoid storing duplicated solutions, we have $O(2^{2^{O(\tw^2 \log s)}})$ solutions. Hence, when building the solution set of $X$ from its child or children, we process at most these many elements. The size of a solution $S$ of $X$ is $O(2^{O(\tw^2 \log s)})$, and each extension takes polynomial time in the size of the extended solution. Let $f(\tw,s)=2^{2^{O(\tw^2 \log s)}}$ and $g(\tw,s)=2^{O(\tw^2 \log s)}$,  and observe that $g(\tw,s)^{O(1)} = o(f(\tw,s))$. It follows that constructing the solution set of $X$ takes $O(f(\tw,s) \cdot g(\tw,s)^{O(1)}) = O(f(\tw,s))=O(2^{2^{O(\tw^2 \log s)}})$. In addition, if $X$ is a forget bag, we also need to remove possible duplicated solutions. This can be done by sorting the solutions, which takes $O(g(\tw,s)) \times O(f(\tw,s) \log f(\tw,s))$. Also, since $O(g(\tw,s)) = O(\log f(\tw,s))$, we have $O(f(\tw,s) \log f(\tw,s))= O(f(\tw,s) \cdot g(\tw,s)) =O(f(\tw,s))$. The statement then follows because there are $O(n)$ bags.
\end{proof}

\section{Discussion and Open Problems}
We have proved that the \clubedel{s} problem parameterized by $s+\tw$ (where $\tw$ bounds the treewidth of the input graph) belongs to \FPT. On the other hand, we know that  the problem parameterized by $s$ alone is para\NP-hard. It remains open the complexity of \clubedel{s} parameterized by $\tw$ alone. It is not difficult to see that our approach can be slightly modified to solve a similar variant of \clubedel{s}, namely \clubvdel{s}, in which we seek for $k$ vertices whose removal yields a set of disjoint $s$-clubs. 
Concerning the parameter $k$, it would be interesting to know whether \clubedel{s} is in \FPT~when parameterized by $s+k$. Finally, we observe that it is possible to adjust a reduction in~\cite{DBLP:conf/fct/DondiL19} to prove that \clubedel{s} remains para\NP-hard even when parameterized by $s+d$, where $d$ is the number of clusters in the sought solution. With respect to parameter $d$, we also know that the problem has no subexponential-time parameterized algorithm in $k+d$~\cite{DBLP:journals/jcss/MisraPS20}. Yet,  whether $k+d$ is a tractable parameterization is an interesting question.

\clearpage

\bibliographystyle{splncs04}
\bibliography{bibliography}

\end{document}